\documentclass[11pt,english]{article} 
\usepackage[T1]{fontenc}
\usepackage[latin9]{inputenc}
\usepackage{hyperref}
\usepackage{verbatim}
\usepackage{amsfonts}
\usepackage{amsmath,amssymb,amsthm}
\usepackage{multirow}
\usepackage{url}

\usepackage{amsfonts}
\usepackage{courier}
\usepackage{fullpage}

\usepackage{epsfig}
\usepackage[usenames]{color}

\usepackage{bm}
\usepackage{times}

\newtheorem{theorem}{Theorem}[section]
\newtheorem{lemma}[theorem]{Lemma}

\theoremstyle{definition}
\newtheorem{definition}[theorem]{Definition}

\renewcommand{\paragraph}[1]{\noindent {\bf #1}}
\newcommand{\Exp}{\operatornamewithlimits{\mathbb{E}}}

\newcommand{\eps}{\epsilon}

\newcommand{\R}{\mathbb{R}}

\newcommand{\eqdef}{{\stackrel{\rm def}{=}}}

\newcommand{\polylog}[1]{\mathrm{polylog}{#1}}
\newcommand{\poly}[1]{\mathrm{poly}(#1)}

\newcommand{\mcA}{{\mathcal A}}

\newcommand{\ith}[1]{{#1}^{\text{th}}}
\newcommand{\supp}[1]{\mathrm{supp}(#1)}

\newcommand{\parent}[1]{\mathrm{parent}(#1)}

\newcommand{\ignore}[1]{}

\makeatother

\usepackage{babel}
\makeatother

\begin{document}

\title{Space-efficient Local Computation Algorithms}
\author{
Noga Alon\thanks{Sackler School of Mathematics and Blavatnik School of Computer Science,
Tel Aviv University, Tel Aviv 69978, Israel and Institute for Advanced Study,
Princeton, New Jersey 08540, USA.  E-mail: {\tt nogaa@tau.ac.il}.
Research supported in part by an ERC Advanced grant, by a USA-Israeli BSF
grant and by NSF grant No.  DMS-0835373.}
\and
Ronitt Rubinfeld\thanks{CSAIL, MIT, Cambridge, MA 02139, USA and
School of Computer Science, Tel Aviv University, Tel Aviv 69978, Israel.
E-mail: {\tt  ronitt@csail.mit.edu}. 
Supported by NSF grants CCF-0728645 and CCF-1065125,
Marie Curie Reintegration grant PIRG03-GA-2008-231077
 and the Israel Science Foundation grant nos. 1147/09 and  1675/09.}
\and
Shai Vardi
\thanks{School of Computer Science, Tel Aviv University, Tel Aviv 69978, Israel.
 E-mail: {\tt  shaivardi@gmail.com}. Supported by
 Israel Science Foundation grant no. 1147/09. }
\and 
Ning Xie
\thanks{CSAIL, MIT, Cambridge MA 02139, USA.
 E-mail: {\tt  ningxie@csail.mit.edu}. 
 Research supported by NSF grants CCF-0728645, CCF-0729011 and CCF-1065125.}
}
\date{}

\setcounter{page}{0}
\maketitle

\begin{abstract}
Recently Rubinfeld et al. (ICS 2011, pp. 223--238) proposed a new model of sublinear algorithms 
called \emph{local computation algorithms}. 
In this model, a computation problem $F$ may have more than one legal solution and 
each of them consists of many bits.
The local computation algorithm for $F$ should answer in an online fashion, for any index $i$, 
the $i^{\mathrm{th}}$ bit of some legal solution of $F$. 
Further, all the answers given by the algorithm should be consistent with at least
one solution of $F$.  

In this work, we continue the study of local computation algorithms.
In particular, we develop a technique which under certain conditions can be applied 
to construct local computation algorithms that run not only in polylogarithmic time but also 
in polylogarithmic \emph{space}.
Moreover, these local computation algorithms are easily parallelizable and 
can answer all parallel queries consistently.
Our main technical tools are pseudorandom numbers with bounded independence
and the theory of branching processes.
\end{abstract}

\newpage

\section{Introduction}\label{Sec:intro}

The classical view of algorithmic analysis, in which the algorithm
reads the entire input, performs a computation and then writes out the entire output, 
is less applicable in the context of computations on massive data sets.    
To address this difficulty, several alternative models of computation have been adopted, 
including distributed
computation as well as  various sub-linear time and space  models.

Local computation algorithms (LCAs) were proposed in~\cite{RTVX11}
to model the scenario in which inputs to and outputs
from the algorithms are large,
such that writing out the entire output
requires an amount of time that is unacceptable.   On the other
hand, only
small portions of the  output are required at any point in time
by any specific user.
LCAs support queries to the output by the user, such
that after each query to a specified location $i$, the
LCA outputs the value of the output at location $i$.    
LCAs were inspired by and intended as a generalization of several
models that appear in the literature, including local algorithms, locally decodable
codes and local reconstruction algorithms.
LCAs whose time complexity is efficient in terms of the amount
of solution requested by the user
have been given for various combinatorial and
coding theoretic problems.

One difficulty
is that for many computations, more than one output is considered to
be valid, yet the values returned by the LCA over time must be consistent.  
Often, the straightforward solutions ask that the LCA store 
intermediate values of the computations in order to
maintain consistency for later computations.
Though standard techniques can be useful for recomputing
the values of random coin tosses in a straightforward manner,
some algorithms (e.g., many greedy algorithms) choose very different solutions 
based on the order of input queries.    
Thus, though the time requirements of the LCA may be efficient
for each query, it is not always clear how to {\em bound the storage requirements}
of the LCA by a function that is sublinear in the size of the query history.
It is this issue that we focus on in this paper.


\subsection{Our main results}
Before stating our main results, we mention two additional
desirable properties of LCAs.     Both of these properties are
achieved in our constructions of LCAs with small storage requirements.
The first is that an LCA should be {\em query oblivious},
that is the outputs of $\mcA$ should not depend on the order of the queries but
only on the input and the random bits generated on the random tape of $\mcA$.
The second is that the LCA should be {\em parallelizable}, i.e., that it
is able to answer multiple queries simultaneously in a consistent manner.

All the LCAs given in~\cite{RTVX11b} suffer from one or more of the following drawbacks: 
the worst case space complexity is linear, 
the LCA is not query oblivious, 
and the LCA is not parallelizable.
We give new techniques to construct LCAs for the problems studied in~\cite{RTVX11b} which
run in polylogarithmic time as well as polylogarithmic space. 
Moreover, all of the LCAs are query oblivious and easily parallelizable.

\begin{theorem}[Main Theorem $1$ (informal)]
There is an LCA for \emph{Hypergraph Coloring} that runs in polylogarithmic time and space.
Moreover, the LCA is query oblivious and parallelizable.
\end{theorem}

\begin{theorem}[Main Theorem $2$ (informal)]
There is an LCA for \emph{Maximal Independent Set} that runs in polylogarithmic time and space.
Moreover, the LCA is query oblivious and parallelizable.
\end{theorem}

We remark that following \cite{RTVX11b}, 
analogous techniques can be applied to construct LCAs with
all of the desirable properties for the
radio network problem and $k$-CNF problems. 

%
%

\subsection{Techniques}
There are two main technical obstacle in making the LCAs constructed in~\cite{RTVX11b} space efficient, 
query oblivious and parallelizable.
The first is that LCAs need to remember all the random bits used in computing previous queries.   
The second issue is more subtle --
\cite{RTVX11b} give LCAs based on algorithms which use very little additional time resources per query as
they simulate greedy algorithms.   
These LCAs output results that depend directly on the orders in which queries are fed into the algorithms.   



We address the randomness issue first.
The space inefficient LCAs constructed in~\cite{RTVX11b} for the problems of concern to us are probabilistic by nature.
Consistency among answers to the queries seems to demand that the algorithm 
keeps track of all random bits used so far, which would incur linear space complexity.
A simple but very useful observation is that all the computations
are local and thus involve a very small number of random bits.
Therefore we may replace the truly random bits with random variables
of limited independence. The construction of small sample space $k$-wise 
independent random variables of Alon et al.~\cite{ABI86} allows
us to reduce the space complexity from linear to polylogarithmic.   
This allows us to prove our main theorem on the LCA for the maximal independent set problem.   
It is also an important ingredient in constructing our LCA for \emph{Hypergraph Coloring}.
We believe such a technique will be a standard tool in future development of LCAs.

For \emph{Hypergraph Coloring}, we need to also address the second issue raised above.
The original LCA for \emph{Hypergraph Coloring} in~\cite{RTVX11b} emulates Alon's algorithm~\cite{Alo91}.
Alon's algorithm runs in three phases. During the first phase,
it colors \emph{all} vertices in an \emph{arbitrary} order.
Such an algorithm looks ``global'' in nature and it is therefore non-trivial to turn it into an LCA.
In~\cite{RTVX11b}, they use the order of vertices being queried as the order of coloring
in Alon's algorithm, hence the algorithm needs to store all answers to previous queries
and requires linear space in computation.

We take a different approach to overcome this difficulty.
Observe that there is some ``local'' dependency among the colors of vertices -- namely, 
the color of any vertex depends only on the colors of at most a constant number, say $D$, other vertices. 
The colors of these vertices in turn depend on the colors of their neighboring vertices, and so on.
We can model the hypergraph coloring process by a query tree:
Suppose the color of vertex $x$ is now queried.
Then the root node of the query tree is $x$, the nodes on the first level are the vertices
whose colors the color of $x$ depends on. In general, the colors of nodes on level $i$ depends on\footnote{In fact,
they may depend on the colors of some nodes on levels lower than $i$. However, as we care only about query complexity,
we will focus on the worst case that the query relations form a tree.}
the colors of nodes on level $i+1$. 
Note that the query tree has degree bound $D$ and moreover,
the size of the query tree clearly depends on the order in which vertices are colored,
since the color of a vertex depends only on vertices that are colored before it.
In particular, if $x$ is the $\ith{k}$ vertex to be colored, then the query tree contains at most $k$ vertices.

An important fact to note is that Alon's algorithm works for \emph{any} order, in particular, it works for a random order.
Therefore we can apply the random order method of Nguyen and Onak~\cite{NO08}:
generate a random number $r\in [0,1]$, called the \emph{rank}, and use these ranks to prune the original query tree
into a \emph{random query tree $\mathcal{T}$}.
Specifically, $\mathcal{T}$ is defined recursively: 
the root of $\mathcal{T}$ is still $x$. A node $z$ is in $\mathcal{T}$ if its parent node $y$
in the original query tree is in $\mathcal{T}$ and $r(z)<r(y)$.  
Intuitively, a random query tree is small and indeed it is surprisingly small~\cite{NO08}:
the expected size of $\mathcal{T}$ is $\frac{e^D -1}{D}$, a constant! 

Therefore, if we color the vertices in the hypergraph in a random order, the \emph{expected}
number of vertices we need to color is only a constant. 
However, such an ``average case'' result is insufficient for our LCA purpose: what
we need is a ``worst case'' result which tells almost surely how large a random query tree will be.
In other words, we need a concentration result on the sizes of the random query trees.
The previous techniques in~\cite{NO08, YYI09} do not seem to work in this setting.

Consider the worst case in which the rank of the root node $x$ is $1$.
A key observation is, although there are $D$ child nodes of $x$, only the nodes whose ranks are 
close to $1$ are important, as the child nodes with smaller ranks will die out quickly.
But in expectation there will be very few important nodes!
This inspires us to partition the random query tree into $D+1$ levels based on the ranks of the nodes,
and analyze the sizes of trees on each level using the theory of branching processes.
In particular, we apply a quantitative bound on the total number of off-springs of a Galton-Watson process~\cite{Ott49}
to show that, for any $m>0$, with probability at least $1-1/m^{2}$ 
the size of a random query tree has at most $C(D)\log^{D+1}m$ vertices, where
$C(D)$ is some constant depending only on $D$.
We conjecture that the upper bound can be further reduced to $C(D)\log{m}$.

However, the random order approach raise another issue: how do we store the ranks of all vertices?
Observe that in constructing a random query tree, the actual values of the ranks are never used --
only the relative orders between vertices matter. 
This fact together with the fact that all computations are local enables us to 
replace the rank function with some pseudorandom ordering among the vertices,
see Section~\ref{Sec:orderings} for formal definition and construction.
The space complexity of the pseudorandom ordering is only polylogarithmic,
thus making the total space complexity of the LCA also polylogarithmic.

\subsection{Other related work}
{\em Locally decodable
codes} \cite{KT00} which given an encoding of a message,
provide quick access to the requested
bits of the original message, can be viewed as
LCAs.   Known constructions of LDCs
are efficient and use small space \cite{Yek11}.   LCAs generalize
the {\em reconstruction} models
described in \cite{ACC+08, CS06a,SS10, JR11}.  These models
describe scenarios where an input string that has a 
certain property, such as monotonicity, is assumed
to be corrupted at a relatively small number of
locations.  The reconstruction algorithm gives fast query access
to an uncorrupted version of the string that is close to the original
input.   Most of the works mentioned are also efficient in terms of space.

In \cite {RTVX11}, it is noted that  the model of LCAs is related to {\em local algorithms}, 
studied in the context of distributed computing \cite{NS95, MNS95, KMN+05, KMW06, KW06, KM07, K09}.  
This is due to a reduction given by Parnas Ron \cite{PR07} which allows one to 
construct (sequential) LCAs based on constant round distributed algorithms.   
Note that this relationship does not immediately yield space-efficient
local algorithms, nor does it yield sub-linear time LCAs when 
used with parallel or distributed 
algorithms whose round complexity is $O(\log{n})$.

Recent exciting developments in sublinear time algorithms for sparse graph and
combinatorial optimization problems have led to new constant time
algorithms for approximating the size of a minimum vertex cover, maximal matching, 
maximum matching, minimum dominating set, minimum set cover, packing 
and covering problems (cf. \cite{ PR07, MR06, NO08, YYI09}).
For example, for \emph{Maximal Independent Set},
these algorithms construct a constant-time oracle which for most,
but not all, vertices outputs 
whether or not the vertex is part of the independent set.  
For the above approximation algorithms, it is not necessary
to get the correct answer for each vertex, but for LCAs, which
must work for any sequence of online inputs, the requirements are
more stringent, thus the techniques are not applicable without modification.

\subsection{Organization}
The rest of the paper is organized as follows.
Some preliminaries and notations that we use throughout the paper
appear in Section~\ref{Sec:prelim}.
We then prove our main technical result, namely the bound on the sizes of random query trees
in Section~\ref{Sec:query_tree}. 
In Section~\ref{Sec:orderings} we construct pseudorandom orderings with small space.
Finally we apply the techniques developed in Section~\ref{Sec:query_tree} and Section~\ref{Sec:orderings} 
to construct LCAs for the hypergraph coloring problem
and the maximal independent set problem 
in Section~\ref{Sec:hypergraph} and Section~\ref{Sec:mis}, respectively.

\section{Preliminaries}\label{Sec:prelim}

Unless stated otherwise, all logarithms in this paper are to the base $2$.
Let $n\geq 1$ be a natural number. We use $[n]$ to denote the set $\{1,\ldots, n\}$.

All graphs in this paper are undirected graphs. Let $G=(V, E)$ be a graph.
The \emph{distance} between two vertices $u$ and $v$ in $V(G)$, denoted by
$d_{G}(u, v)$, is the length of a shortest path between the two vertices.
We write $N_{G}(v)=\{u\in V(G): (u,v)\in E(G)\}$ to denote
the neighboring vertices of $v$.   Furthermore, 
let $N^{+}_{G}(v)=N(v)\cup \{v\}$.
Let $d_{G}(v)$ denote the degree of a vertex $v$.

\subsection{Local computation algorithms}
We present our model of local computation algorithms:
Let $F$ be a computational problem and $x$ be an input to $F$.
Let $F(x) = \{y~|~ y {\rm ~is~ a~} {\rm valid~solution~}{\rm for~input~} x\}$.
The {\em search problem} is to find any $y \in F(x)$. 
\begin{definition}[$(t, s, \delta)$-local algorithms~\cite{RTVX11b}]
Let $x$ and $F(x)$ be defined as above.
A {\em $(t(n), s(n), \delta(n))$-local computation algorithm} $\mcA$
is a (randomized) algorithm which
implements query access to an arbitrary $y \in F(x)$ and satisfies the following:
$\mcA$ gets a sequence of queries $i_1,\ldots,i_q$ for any $q>0$
and after each query $i_j$ it must produce an output $y_{i_j}$
satisfying that the outputs $y_{i_1},\ldots, y_{i_q}$ are substrings of some $y \in F(x)$.
The probability of success over all $q$ queries must be at least $1-\delta(n)$.
$\mcA$ has access to a random tape and local computation memory on which
it can perform current computations as well as store and retrieve information from previous computations.
We assume that the input $x$, the local computation tape and any
random bits used are all presented in the RAM word model, i.e., 
$\mcA$ is given the ability to access a word of any of these in one step.
The running time of $\mcA$ on any query is at most $t(n)$, which is sublinear in $n$,
and the size of the local computation memory of $\mcA$ is at most $s(n)$.
Unless stated otherwise, we always assume that the error parameter 
$\delta(n)$ is at most some constant, say, $1/3$.
We say that  $\mcA$ is a {\em strongly local computation algorithm}
if both $t(n)$ and $s(n)$ are upper bounded by $\log^{c}n$ for some constant $c$.
\end{definition}

Two important properties of LCAs are as follows:
\begin{definition}[Query oblivious\cite{RTVX11b}]
We say an LCA $\mcA$ is \emph{query order oblivious} (\emph{query oblivious} for short)
if the outputs of $\mcA$ do not depend on the order of the queries but
depend only on the input and the random bits generated on the random tape of $\mcA$.
\end{definition}

\begin{definition}[Parallelizable\cite{RTVX11b}]
We say an LCA $\mcA$ is \emph{parallelizable} if $\mcA$ 
supports parallel queries, that is the LCA is able to answer multiple queries 
simultaneously so that all the answers are consistent.
\end{definition}
\subsection{\texorpdfstring{$k$-wise independent random variables}{k-wise independent random variables}}
\label{Sec:kwi}
Let $1\leq k \leq n$ be an integer.
A distribution $D:\{0,1\}^{n} \to \R^{\geq 0}$ is \emph{$k$-wise independent}
if restricting $D$ to any index subset $S\subset [n]$ of size at most $k$ 
gives rise to a uniform distribution.
A random variable is said to be $k$-wise independent if its distribution function is $k$-wise independent.
Recall that the support of a distribution $D$, denoted $\supp{D}$,
is the set of points at which $D(x)>0$.
We say a discrete distribution $D$ is \emph{symmetric}
if $D(x)=1/|\supp{D}|$ for every $x\in \supp{D}$.
If a distribution $D:\{0,1\}^{n} \to \R^{\geq 0}$ is symmetric with 
$|\supp{D}|\leq 2^{m}$ for some $m\leq n$, then we may index the elements in the support of $D$ 
by $\{0,1\}^{m}$ and call $m$ the \emph{seed length} of the random variable whose distribution is $D$.
We will need the following construction of $k$-wise independent random variables
over $\{0,1\}^{n}$ with small symmetric sample space.

\begin{theorem}[\cite{ABI86}]\label{thm:ABI86}
For every $1\leq k \leq n$, there exists a symmetric distribution $D:\{0,1\}^{n} \to \R^{\geq 0}$
of support size at most $n^{\lfloor \frac{k}{2} \rfloor}$ and is $k$-wise independent.
That is, there is a $k$-wise independent random variable $x=(x_1, \ldots, x_n)$ 
whose seed length is at most $O(k\log n)$. Moreover, for any $1\leq i \leq n$,
$x_i$ can be computed in space $O(k\log n)$.
\end{theorem}

\section{Bounding the size of a random query tree}\label{Sec:query_tree}

\subsection{The problem and our main result}
Consider the following scenario which was first studied by~\cite{NO08}
in the context of constant-time approximation algorithms for maximal matching and some other problems.
We are given a graph $G=(V, E)$ of bounded degree $D$.
A real number $r(v)\in [0,1]$ is assigned independently and uniformly at random 
to every vertex $v$ in the graph.
We call this random number the \emph{rank} of $v$.
Each vertex in the graph $G$ holds an input $x(v) \in R$, where the range $R$ is some finite set.
A randomized Boolean function $F$ is defined inductively on
the vertices in the graph such that $F(v)$ is a function of the input $x(v)$ at $v$
as well as the values of $F$ at the neighbors $w$ of $v$ for which $r(w)<r(v)$. 
The main question is, in order to compute $F(v_{0})$ for any vertex $v_{0}$ in $G$,
how many queries to the inputs of the vertices in the graph are needed?

Here, for the purpose of upper bounding the query complexity,
we may assume for simplicity that the graph $G$ is $D$-regular and furthermore,
$G$ is an infinite $D$-regular tree rooted at $v_{0}$.
It is easy to see that making such modifications to $G$ 
can never decrease the query complexity of computing $F(v_{0})$.

Consider the following question.
We are given an infinite $D$-regular tree $\mathcal{T}$ rooted at $v_{0}$.
Each node $w$ in $\mathcal{T}$ is assigned independently and uniformly at random a real number
$r(w)\in [0,1]$.
For every node $w$ other than $v_{0}$ in $\mathcal{T}$,
let $\parent{w}$ denote the parent node of $w$.
We grow a (possibly infinite) subtree $T$ of $\mathcal{T}$ rooted at $v$ as follows:
a node $w$ is in the subtree $T$ if and only if $\parent{w}$ is in $T$ and 
$r(w)<r(\parent{w})$ (for simplicity we assume all the ranks are distinct real numbers).
That is, we start from the root $v$, add all the children of $v$ whose ranks are smaller than that
of $v$ to $T$. We keep growing $T$ in this manner where a node $w'\in T$ is a leaf node in $T$ 
if the ranks of its $D$ children are all larger than $r(w')$. 
We call the random tree $T$ constructed in this way a \emph{query tree} and 
we denote by $|T|$ the random variable that corresponds to the size of $T$.
We would like to know what are the typical values of $|T|$.

Following~\cite{NO08, Ona10}, we have that, for any node $w$ that is
at distance $t$ from the root $v_{0}$,  
$\Pr[w\in T]$=1/(t+1)! as such an event happens if and only if the ranks of the
$t+1$ nodes along the shortest path from $v_{0}$ to $w$ is in monotone decreasing order.
It follows from linearity of expectation that the expected value of $|T|$ is 
given by the elegant formula~\cite{Ona10}
\[
\Exp[|T|]=\sum_{t=0}^{\infty}\frac{D^t}{(t+1)!}=\frac{e^D -1}{D},
\]
which is a constant depending only on the degree bound $D$.

Our main result in this section can be regarded as showing that in fact $|T|$ is highly concentrated 
around its mean:

\begin{theorem}\label{thm:query_tree}
For any degree bound $D\geq 2$, there is a constant $C(D)$ which depends on $D$ only such that 
for all large enough $N$,
\[
\Pr[|T|>C(D)\log^{D+1}N]<1/N^2.
\]
\end{theorem}

\subsection{Breaking the query tree into levels}
A key idea in the proof is to break the query tree into levels and then upper bound the 
sizes of the subtrees on each level separately.
First partition the interval $[0,1]$ into $D+1$ sub-intervals:
$I_{i}:=(1-\frac{i}{D+1}, 1-\frac{i-1}{D+1}]$ for $i=1,2, \ldots, D$ and $I_{D+1}=[0, \frac{1}{D+1}]$.
We then decompose the query tree $T$ into $D+1$ levels such that
a node $v\in T$ is said to be on level $i$ if $r(v)\in I_{i}$.
For ease of exposition, in the following we consider the worst case that $r(v_{0})\in I_{1}$.
Then the vertices of $T$ on level $1$ form a tree which we call $T_{1}=T_{1}^{(1)}$ rooted at $v_{0}$.
The vertices of $T$ on level $2$ will in general form a set of trees 
$\{T_{2}^{(1)}, \ldots, T_{2}^{(m_2)} \}$, 
where the total number of such trees $m_{2}$ 
is at most $D$ times the number of nodes in $T_{1}$ (we have
only inequality here because some of the child nodes in $\mathcal{T}$ of the nodes
in $T_{1}$ may fall into levels $2$, $3$, etc).
Finally the nodes on level $D+1$ form a forest 
$\{T_{D+1}^{(1)}, \ldots, T_{D+1}^{(m_{\scriptscriptstyle{D+1}})} \}$.
Note that all these trees $\{T_{i}^{(j)}\}$ are generated by the same stochastic process,
as the ranks of all nodes in $\mathcal{T}$ are i.i.d. random variables. 
The next lemma shows that  
each of the subtrees on any level is of size $O(\log N)$ with probability at least $1-1/N^3$,

\begin{lemma}\label{lemma:tree_size}
For any $1\leq i \leq D+1$ and any $1\leq j \leq m_{i}$,
with probability at least $1-1/N^3$, $|T_{i}^{(j)}|=O(\log N)$.
\end{lemma}

One can see that Theorem~\ref{thm:query_tree} follows directly from Lemma~\ref{lemma:tree_size}:
Once again we consider the worst case that $r(v_{0})\in I_{1}$.
By Lemma~\ref{lemma:tree_size}, the size of $T_{1}$
is at most $O(\log N)$ with probability at least $1-1/N^3$.
In what follows, we always condition our argument upon that
this event happens.
Notice that the root of any tree on level $2$ must have some node in $T_{1}$
as its parent node; it follows that $m_2$, 
the number of trees on level $2$, is at most $D$ times the size of $T_{1}$, 
hence $m_{2}=O(\log N)$.
Now applying Lemma~\ref{lemma:tree_size} to each of the $m_{2}$ trees on level $2$
and assume that the high probability event claimed in Lemma~\ref{lemma:tree_size} 
happens in each of the subtree cases,
we get that the total number of nodes at level $2$ is at most $O(\log^{2}N)$.
Once again, any tree on level $3$ must have some node in either level $1$ or level $2$
as its parent node, 
so the total number of trees on level $3$ is also at most 
$D(O(\log N)+O(\log^{2}N))=O(\log^{2}N)$.
Applying this argument inductively, 
we get that $m_{i}=O(\log^{i-1}N)$ for $i=2,3,\ldots, D+1$.
Consequently, the total number of nodes at all $D+1$ levels is at most
$O(\log N) + O(\log^{2}N)+ \cdots + O(\log^{D+1}N)=O(\log^{D+1}N)$,
assuming the high probability event in Lemma~\ref{lemma:tree_size}
holds for all the subtrees in all the levels.
By the union bound, this happens with probability at least $1-O(\log^{D+1}N)/N^{3}>1-1/N^{2}$,
thus proving Theorem~\ref{thm:query_tree}.

The proof of Lemma~\ref{lemma:tree_size} requires results in
branching processes, in particular the Galton-Watson processes.

\subsection{Galton-Watson processes}
Consider a Galton-Watson process defined by the probability function
$\mathbf{p}:=\{p_{k}; k=0,1,2, \ldots \}$, with $p_{k}\geq 0$ and $\sum_{k}p_{k}=1$.
Let $f(s)=\sum_{k=0}^{\infty}p_{k}s^{k}$
be the generating function of $\mathbf{p}$.
For $i=0, 1, \ldots,$ let $Z_{i}$ be the number of off-springs in the $\ith{i}$ generation.
Clearly $Z_{0}=1$ and $\{Z_{i}: i=0, 1, \ldots \}$ form a Markov chain.
Let $m:=\Exp[Z_{1}]=\sum_{k}kp_{k}$ be the expected number of children of any individual.
The classical result of the Galton-Watson processes
is that the \emph{survival probability} (namely $\lim_{n\to \infty}\Pr[Z_{n}>0]$)
is zero if and only if $m\leq 1$.
Let $Z=Z_{0}+Z_{1}+\cdots$ be the sum of all off-springs in all generations of the Galton-Watson process.
The following result of Otter is useful in bounding the probability that $Z$ is large.

\begin{theorem}[\cite{Ott49}]\label{thm:Otter}
Suppose $p_{0}>0$ and that there is a point $a>0$ within the circle of convergence of $f$
for which $af'(a)=f(a)$. Let $\alpha=a/f(a)$. 
Let $t=\mathrm{gcd}\{r: p_{r}>0\}$, where $\mathrm{gcd}$ stands for greatest common divisor.
Then

\begin{align} \label{eqn:Otter}
\Pr[Z=n]=\begin{cases}
  t\left(\frac{a}{2\pi \alpha f''(a)} \right)^{1/2}\alpha^{-n}n^{-3/2} +O(\alpha^{-n}n^{-5/2}), 
    &  \text{if $n \equiv 1 \pmod{t}$;} \\
  0, &  \text{if $n \not\equiv 1 \pmod{t}$.}
\end{cases}
\end{align}
\end{theorem}

In particular, if the process is \emph{non-arithmetic}, i.e. $\mathrm{gcd}\{r: p_{r}>0\}=1$,
and $\frac{a}{\alpha f''(a)}$ is finite, then
\[
\Pr[Z=n]=O(\alpha^{-n}n^{-3/2}),
\]
and consequently $\Pr[Z\geq n]=O(\alpha^{-n})$.

\subsection{Proof of Lemma~\ref{lemma:tree_size}}
To simplify exposition, we prove Lemma~\ref{lemma:tree_size} for the case of tree $T_{1}$.
Recall that $T_{1}$ is constructed recursively as follows: 
for every child node $v$ of $v_{0}$ in $\mathcal{T}$, we add $v$ to $T_{1}$ if
$r(v)<r(v_{0})$ and $r(v) \in I_{1}$.
Then for every child node $v$ of $v_{0}$ in $T_{1}$, we add the child node $w$ of $v$ in $\mathcal{T}$
to $T_{1}$ if $r(w)<r(v)$ and $r(w) \in I_{1}$.
We repeat this process until there is no node that can be added to $T_{1}$.

Once again, we work with the worst case that $r(v_{0})=1$.
To upper bound the size of $T_{1}$, we consider a related random process which also 
grows a subtree of $\mathcal{T}$ rooted at $v_{0}$, and denote it by $T'_{1}$.
The process that grows $T'_{1}$ is the same as that of $T_{1}$ 
except for the following difference: if $v\in T'_{1}$ and $w$ is a child node of $v$ in $\mathcal{T}$, 
then we add $w$ to $T'_{1}$ as long as $r(w) \in I_{1}$, 
but give up the requirement that $r(w)<r(v)$.  
Clearly, we always have $T_{1}\subseteq T'_{1}$ and hence $|T'_{1}|\geq |T_{1}|$.
 
Note that the random process that generates $T'_{1}$ is in fact a Galton-Watson process,
as the rank of each node in $\mathcal{T}$ is independently and uniformly distributed in $[0,1]$.
Since $|I_{1}|=1/(D+1)$, the probability function is
\[
\mathbf{p}=\{(1-q)^D, \binom{D}{1}q(1-q)^{D-1}, \binom{D}{2}q^{2}(1-q)^{D-2}, \ldots, q^{D}\},
\]
where $q:=1/(D+1)$ is the probability that a child node in $\mathcal{T}$
appears in $T'_{1}$ when its parent node is in $T'_{1}$.
Note that the expected number of children of a node in $T'_{1}$ is $Dq=D/(D+1)<1$,
so the tree $T'_{1}$ is a finite tree with probability one.

The generating function of $\mathbf{p}$ is 
\[
f(s)=(1-q+qs)^{D},
\]
as the probability function $\{p_{k}\}$ obeys the binomial distribution $p_{k}=b(k,D,q)$.
In addition, the convergence radius of $f$ is $\rho=\infty$ 
since $\{p_{k}\}$ has only a finite number of non-zero terms.

Solving the equation $af'(a)=f(a)$ yields $a=\frac{1-q}{q(D-1)}=\frac{D}{D-1}$.
It follows that (since $D\geq 2$)
\[
f''(a)=q^{2}D(D-1)\left(1-q+\frac{1-q}{D-1}\right)^{D-2}>0,
\]
hence the coefficient in $(\ref{eqn:Otter})$ is non-singular.

Let $\alpha(D):=a/f(a)=1/f'(a)$, then 
\begin{align*}
1/\alpha(D) =f'(a)
&=\frac{D}{D+1}(\frac{D^2}{D^{2}-1})^{D-1} \\
&=(1+\frac{1}{D^{2}-1})^{(D^{2}-1)/(D+1)}\frac{D}{D+1} \\
&<e^{1/(D+1)}\frac{D}{D+1} \\
&<\left((1+\frac{1}{D})^{D+1}\right)^{1/(D+1)}\frac{D}{D+1} \\
&=1,
\end{align*}
where in the third and the fourth steps we use
the inequality (see e.g. ~\cite{Mis70}) that $(1+\frac{1}{t})^{t} <e < (1+\frac{1}{t})^{t+1}$
for any positive integer $t$.
This shows that $\alpha(D)$ is a constant greater than $1$. 

Now applying Theorem~\ref{thm:Otter} to the Galton-Watson process
which generates $T'_{1}$ (note that $t=1$ in our case) gives that, for all large enough $n$,
$\Pr[|T'_{1}|=n]\leq 2^{-cn}$ for some constant $c$.
It follows that $\Pr[|T'_{1}|\geq n]\leq \sum_{i=n}^{\infty}2^{-ci}\leq 2^{-\Omega(n)}$
for all large enough $n$.
Hence for all large enough $N$, with probability at least $1-1/N^3$, $|T_{1}|\leq |T'_{1}|=O(\log N)$.
This completes the proof of Lemma~\ref{lemma:tree_size}.

\section{\texorpdfstring{Construction of almost $k$-wise independent random orderings}
{Construction of almost k-wise independent random orderings}}
\label{Sec:orderings}
An important observation that enables us to make some of our local algorithms run
in polylogarithmic space is the following. 
In the construction of a random query tree $\mathcal{T}$,
we do not need to generate a random real number $r(v)\in [0,1]$
independently for each vertex $v\in \mathcal{T}$; instead only the \emph{relative orderings}
among the vertices in $\mathcal{T}$ matter. 
Indeed, when generating a random query tree, we only compare the
ranks between a child node $w$ and its parent node $v$ to see if $r(w)<r(v)$;
the absolute values of $r(w)$ and $r(v)$ are irrelevant and 
are used only to facilitate our analysis in Section~\ref{Sec:query_tree}.
Moreover, since (almost surely) all our computations in the local algorithms involve only
a very small number of, say at  most $k$, vertices, 
so instead of requiring a random source that 
generates total independent random ordering among all nodes in the graph,
any pseudorandom generator that produces \emph{$k$-wise independent random ordering} 
suffices for our purpose. We now give the formal definition of such orderings.

Let $m\geq 1$ be an integer. 
Let $\mathcal{D}$ be any set with $m$ elements.
For simplicity and without loss of generality,
we may assume that $\mathcal{D}=[m]$.
Let $\mathcal{R}$ be a totally ordered set.
An \emph{ordering} of $[m]$ is an injective function $r: [m] \to \mathcal{R}$.
Note that we can \emph{project} $r$ to an element 
in the symmetric permutation group $\mathcal{S}_{m}$ in a natural way:
arrange the elements $\{r(1), \ldots, r(m)\}$ in $\mathcal{R}$
in the monotone increasing order and call the permutation of $[m]$ corresponding to
this ordering the \emph{projection of $r$} onto $\mathcal{S}_{m}$ and denote it by $P_{\mathcal{S}_{m}}r$.
In general the projection $P_{\mathcal{S}_{m}}$ is not injective.
Let $\mathbf{r}=\{r_{i}\}_{i\in I}$ be any family of orderings indexed by $I$.
The \emph{random ordering $D_{\mathbf{r}}$} of $[m]$ is a distribution over 
a family of orderings $\mathbf{r}$.
For any integer $2\leq k \leq m$, we say
a random ordering $D_{\mathbf{r}}$ is \emph{$k$-wise independent}
if for any subset $S\subseteq [m]$ of size $k$,
the restriction of the projection onto $\mathcal{S}_{m}$ of $D_{\mathbf{r}}$ over $S$ is
uniform over all the $k!$ possible orderings among the $k$ elements in $S$. 
A random ordering $D_{\mathbf{r}}$ is said to \emph{$\eps$-almost $k$-wise independent} if the
statistical distance between $D_{\mathbf{r}}$ is at most $\eps$ from
some $k$-wise independent random ordering.
Note that our definitions of $k$-wise independent random ordering
and almost $k$-wise independent random ordering
are different from that of $k$-wise independent permutation 
and almost $k$-wise independent permutation (see e.g.~\cite{KNR09}),
where the latter requires that the function to be a \emph{permutation}
(i.e., the domain and the range of the function are the same set).
In this section we give a construction of $\frac{1}{m^{2}}$-almost $k$-wise 
independent random ordering whose seed length is $O(k\log^{2}m)$.
In our later applications $k=\polylog{m}$ so the seed length of the almost
$k$-wise independent random ordering is also polylogarithmic.


\begin{theorem}\label{thm:orderings}
Let $m\geq 2$ be an integer and let $2\leq k \leq m$.
Then there is a construction of $\frac{1}{m^{2}}$-almost $k$-wise 
independent random ordering over $[m]$ whose seed length is $O(k\log^{2}m)$.
\end{theorem}

\begin{proof}
For simplicity we assume that $m$ is a power of $2$.
Let $s=4\log m$. 
We generate $s$ \emph{independent} copies of
$k$-wise independent random variables $Z_{1}, \ldots, Z_{s}$ 
with each $Z_{\ell}$, $1\leq \ell \leq s$, in $\{0,1\}^{m}$. 
By Theorem~\ref{thm:ABI86}, the seed length of each
random variable $Z_{\ell}$ is $O(k\log m)$ and therefore the total space needed to store these
random seeds is $O(k\log^{2}m)$.
Let these $k$-wise independent $m$-bit random variables be
\begin{align*} 
Z_{1}&=z_{1,1}, \ldots, z_{1,m}; \\
Z_{2}&=z_{2,1}, \ldots, z_{2,m}; \\
&\ldots   \ldots \\ 
Z_{s}&=z_{s,1}, \ldots, z_{s,m}.
\end{align*}

Now for every $1\leq i \leq m$, we view each $r(i)\eqdef z_{1,i}z_{2,i}\cdots z_{s,i}$ 
as an integer in $\{0,1, \ldots, 2^{s}-1\}$ written in the $s$-bit binary representation 
and use $r:[m] \to \{0,1, \ldots, 2^{s}-1\}$ as the ranking function 
to order the $m$ elements in the set.
We next show that, with probability at least $1-1/m^{2}$,  
$r(1), \ldots, r(m)$ are distinct $m$ integers.

Let $1\leq i < j \leq m$
be any two distinct indices. For every $1\leq \ell \leq s$,
since $z_{\ell,1}, \ldots, z_{\ell,m}$ are $k$-wise independent 
and thus also pair-wise independent, it follows that
$\Pr[z_{\ell,i}=z_{\ell,j}]=1/2$. Moreover, as all $Z_{1}, \ldots, Z_{s}$ are independent,
we therefore have
\begin{align*}
  \Pr[r(i)=r(j)]
&=\Pr[\text{$z_{\ell,i}=z_{\ell,j}$ for every $1\leq \ell \leq s$}] \\
&= \prod_{\ell=1}^{s} \Pr[z_{\ell,i}=z_{\ell,j}] \\
&=(1/2)^{s} \\
&=1/m^{4}.
\end{align*}
Applying a union bound argument over all $\binom{m}{2}$ distinct pairs of indices gives 
that with probability at least $1-1/m^{2}$, all these $m$ numbers are distinct. 

Since each $Z_{\ell}$, $1 \leq \ell \leq s$, is a $k$-wise independent random
variable in $\{0,1\}^{m}$, therefore for any subset $\{i_{1}, \ldots, i_{k}\}$
of $k$ indices, $(r(i_{1}), \ldots, r(i_{k}))$
is distributed uniformly over all $2^{ks}$ tuples. 
By symmetry, conditioned on that $r(i_{1}), \ldots, r(i_{k})$ are all distinct,
the restriction of the ordering induced by the ranking function $r$ to 
$\{i_{1}, \ldots, i_{k}\}$ is completely independent. 
Finally, since the probability that $r(1), \ldots, r(m)$ are not distinct
is at most $1/m^{2}$, it follows that the random ordering induced by $r$ is 
$\frac{1}{m^{2}}$-almost $k$-wise independent.
\end{proof}

\section{LCA for \emph{Hypergraph Coloring}}\label{Sec:hypergraph}
We now apply the technical tools developed in Section~\ref{Sec:query_tree} and
Section~\ref{Sec:orderings} to the design and analysis of LCAs.

Recall that a \emph{hypergraph} $H$ is a pair $H = (V,E)$ where $V$ is a finite set whose elements are
called \emph{nodes} or \emph{vertices}, and $E$ is a family of non-empty subsets of $V$, 
called \emph{hyperedges}. 
A hypergraph is called \emph{$k$-uniform} if each of its
hyperedges contains precisely $k$ vertices.
A \emph{two-coloring} of a hypergraph $H$ is a mapping $\mathbf{c}: V\to \{\text{red, blue}\}$
such that no hyperedge in $E$ is monochromatic.
If such a coloring exists, then we say $H$ is \emph{two-colorable}.
In this paper we assume that each
hyperedge in $H$ intersects at most $d$ other hyperedges.
Let $N$ be the number of hyperedges in $H$. 
Here and after we think of $k$ and $d$ as fixed constants
and all asymptotic forms are with respect to $N$.
By the Lov{\'{a}}sz Local Lemma (see, e.g.~\cite{AS00}) 
when $e(d+1) \leq 2^{k-1}$, the hypergraph $H$ is two-colorable.

Following~\cite{RTVX11b}, we let $m$ be the total number of vertices in $H$. 
Note that $m\leq kN$, so $m=O(N)$.
For any vertex $x\in V$, we use $\mathcal{E}(x)$ to denote the set of hyperedges $x$ belongs to.
For any hypergraph $H = (V,E)$,
we define a \emph{vertex-hyperedge incidence matrix} $\mathcal{M}\in \{0,1\}^{m\times N}$
so that, for every vertex $x$ and every hyperedge $e$, 
$\mathcal{M}_{x,e}=1$ if and only if $e\in \mathcal{E}(x)$.
Because we assume both $k$ and $d$ are constants, 
the incidence matrix $\mathcal{M}$ is necessarily very sparse. 
Therefore, we further assume that the matrix $\mathcal{M}$ is implemented via
linked lists for each row (that is, vertex $x$) and each column (that is, hyperedge $e$). 

Let $G$ be the \emph{dependency graph} of the hyperedges in $H$. 
That is, the vertices of the undirected graph $G$
are the $N$ hyperedges of $H$ and a hyperedge $E_{i}$ is connected to
another hyperedge $E_{j}$ in $G$ if $E_{i}\cap E_{j} \neq \emptyset$.
It is easy to see that if the input hypergraph is given in the 
above described representation, then we can find all the neighbors of any hyperedge $E_{i}$ 
in the dependency graph $G$ (there are at most $d$ of them) in $O(\log N)$ time.

\subsection{Overview of Alon's algorithm}
We now give a sketch of Alon's algorithm~\cite{Alo91}; 
for a detailed description of the algorithm in the context of LCA see~\cite{RTVX11b}.

The algorithm runs in three phases.
In the first phase, we go over all the vertices in the hypergraph in \emph{any} order 
and color them in $\{\mathrm{red}, \mathrm{blue}\}$ uniformly at random.
During this process, if any hyperedge has too many vertices (above some threshold) 
in it are colored in one color and no vertex is colored in the other color, 
then this hyperedge is said to become \emph{dangerous}.
All the uncolored vertices in the dangerous hyperedges are then frozen
and will be skipped during Phase 1 coloring.
A hyperedge is called \emph{survived} if it does not have vertices in both colors
at the end of Phase 1. 
The basic lemma, based on the breakthrough result of Beck~\cite{Bec91},
claims that after Phase 1, almost surely all connected components of the dependency graph $H$
of survived hyperedges are of sizes at most $O(\log{N})$.
We then proceed to the second phase of the algorithm which repeats 
the same coloring process (with some different threshold parameter)
for each connected component and gives rise to connected components of size $O(\log\log{N})$.
Finally in the third phase we perform a brute-force search for a valid coloring whose existence 
is guaranteed by the Lov{\'{a}}sz local lemma.
As each of the connected components in Phase 3 has at most $O(\log\log{N})$ vertices,
the running time of each brute force search is thus bounded by $\polylog{N}$.

To turn Alon's algorithm into an LCA, Rubinfeld et al.~\cite{RTVX11b} note that
one may take the order that vertices are queried as the order to color the vertices
and then in Phase 2 and Phase 3 focus only on the connected components in which the queried vertex lie.
This leads to an LCA with polylogarithmic running time but the space complexity 
can be linear in the worst case (as the algorithm needs to remember the colors
of all previously queried or colored vertices).
In addition, the LCA is not query oblivious and not easily parallelizable.

\subsection{New LCA for \emph{Hypergraph Coloring}}
To remedy these, we add several new ingredients to the LCA in~\cite{RTVX11b} 
and achieve an LCA with both time and space complexity are polylogarithmic.
In addition, the LCA is query oblivious and easily parallelizable.
\vspace{2mm}

\paragraph{1st ingredient: bounded-degree dependency.}
We first make use of the following simple fact:
the color of any fixed vertex in the hypergraph depends only on the colors of a very small number of vertices.
Specifically, if vertex $x$ lies in hyperedges $E_{1}, \ldots, E_{d'}$,
then the color of $x$ depends only on the colors of all the vertices in $E_{1}, \ldots, E_{d'}$.
As every hyperedge is $k$-uniform and each hyperedge
intersects at most $d$ other hyperedges,
the color of any vertex depends on at most the colors of $D=k(d+1)$ other vertices.
\vspace{2mm}

\paragraph{2nd ingredient: random permutation.}
Note that in the first phase of Alon's coloring algorithm, 
\emph{any} order of the vertices will work. 
Therefore, we may apply the idea of random ordering in~\cite{NO08}.
Specifically, suppose we are given a random number generator $r:[m] \to [0,1]$
which assign a random number uniformly and independently to every vertex in the hypergraph.
Suppose the queried vertex is $x$. 
Then we build a (random) query tree $\mathcal{T}$ rooted at $x$ using BFS as follows:
there are at most $D$ other vertices such that the color of $x$ depends on 
the colors of these vertices.
Let $y$ be any of such vertex. If $r(y)<r(x)$, i.e. the random number assigned
to $y$ is smaller than that of $x$, then we add $y$ as a child node of $x$ in $\mathcal{T}$.
We build the query tree this way recursively until there is 
no child node can be added to $\mathcal{T}$. 
By Theorem~\ref{thm:query_tree}, with probability at least $1-1/m^{2}$,
the total number of nodes in $\mathcal{T}$ is at most $\polylog{m}$
and is thus also at most $\polylog{N}$.
This implies that, if we color the vertices in $\mathcal{T}$ in the order from bottom to top (that is, 
we color the leaf nodes first, then the parent nodes of the leaf nodes and so on,
 and color the root node $x$ last), then for any vertex $x$, with probability at least $1-1/m^{2}$
we can follow Alon's algorithm and color at most $\polylog{N}$ vertices (and ignore all other vertices 
in the hypergraph) before coloring $x$.
Therefore the running time of the first phase of our new LCA is (almost surely) at most $\polylog{N}$.
 
\vspace{2mm}

\paragraph{3rd ingredient: $k$-wise independent random ordering.} 
The random permutation method requires linear space to store all the random numbers that have 
been revealed in previous queries in order to make the answers consistent.
However, two useful observations enable us to reduce the space complexity of random ordering
from linear to polylogarithmic.
First, only the relative orderings among vertices matter: in building the query tree $\mathcal{T}$
we only check if $r(y)<r(x)$ but the absolute value of $r(x)$ and $r(y)$ are irrelevant.
Therefore we can replace the random number generator $r$ with an equivalent random ordering function $r\in \mathcal{S}_{m}$,
where $\mathcal{S}_{m}$ is the symmetric group on $m$ elements.
Second, as the query tree size is at most polylogarithmic almost surely, 
the random ordering function $r$ need not be totally random but a polylogarithmic-wise independent permutation suffices\footnote{
Since the full query tree has degree bound $D$, so the total number of nodes queried in building 
the random query tree $\mathcal{T}$ is at most $D|\mathcal{T}|$, which is also at most
polylogarithmic.}.
Therefore we can use the construction in Theorem~\ref{thm:orderings}
of $\frac{1}{m^{2}}$-almost $k$-wise independent random ordering of all the vertices in the 
hypergraph with $k=\polylog{N}$. The space complexity of such a random ordering, or the seed length, 
is $O(k\log^{2}m)=\polylog{N}$.
\vspace{2mm}

\paragraph{4th ingredient: $k$-wise independent random coloring.} 
Finally, the space complexity for storing all the random colors assigned to vertices is also linear in worst case.
Once again we exploit the fact that all computations in LCAs are local to reduce the space complexity.
Specifically, the proof of the basic lemma of Alon's algorithm (see e.g.~\cite[Claim 5.7.2]{AS00}) 
is valid as long as the random coloring of the vertices is $c\log{N}$-wise independent, 
where $c$ is some absolute constant.
Therefore we can replace the truly random numbers in $\{0,1\}^{m}$ used for coloring 
with a $c\log{N}$-wise independent random numbers in $\{0,1\}^{m}$ constructed in Theorem~\ref{thm:ABI86}
thus reducing the space complexity of storing random colors to $O(\log^{2}{N})$.
\vspace{2mm}

\subsection{Pseudocode of the LCA and main result}
To put everything together, we have the following LCA for \emph{Hypergraph Coloring} as illustrated in 
Fig.~\ref{Fig:coloring}, Fig.~\ref{Fig:Phase2} and Fig.~\ref{Fig:Phase3}.
In the preprocessing stage, the algorithm generates $O(\frac{\log{N}}{\log\log{N}})$ copies 
of pseudo-random colors for every vertex in the hypergraph and 
a pseudorandom ordering of all the vertices.
To answer each query, the LCA runs in three phases.
Suppose the color of vertex $x$ is queried.
During the first phase, the algorithm uses BFS to build a random query tree rooted at $x$ and 
then follows Alon's algorithm to color all the vertices in the query tree.
If $x$ gets colored in Phase $1$, the algorithm simply returns that color;
if $x$ is frozen in Phase $1$, then Phase $2$ coloring is invoked.
In the second phase, the algorithm first explores the connected components around $x$ of survived hyperedges.
Then Alon's algorithm is performed again, but this time only on the vertices in the connected component.
For some technical reason, the random coloring process is repeated  
$O(\frac{\log{N}}{\log\log{N}})$ times\footnote{This
is why the algorithm generates many copies of independent pseudorandom colorings at the beginning of the LCA.},
until a  good coloring is found which makes 
all the surviving connected components after Phase $2$ very small.
If $x$ gets colored in the good coloring, then that color is returned; otherwise
the algorithm runs the last phase, 
in which a brute-force search is performed to find the color of $x$.

The time and space complexity as well as the error bound of the LCA are
easy to analyze and we have the following main result of LCA for \emph{Hypergraph Coloring}:
\begin{theorem}\label{thm:hypergraph_main}
Let $d$ and $k$ be such that there exist three positive integers $k_{1}, k_{2}$ and $k_{3}$ such that
the followings hold:
\begin{align*}
k_{1}+k_{2}+{k_3} &=k, \\
16d(d-1)^{3}(d+1) &< 2^{k_{1}},\\
16d(d-1)^{3}(d+1) &< 2^{k_{2}},\\
2e(d+1) &< 2^{k_{3}}.
\end{align*}
Then there exists a  
$(\polylog{N}, \polylog{N}, 1/N)$-local computation algorithm 
which, given a hypergraph $H$ and any sequence of
queries to the colors of vertices $(x_1, x_2, \ldots, x_s)$, 
returns a consistent coloring for all $x_i$'s which 
agrees with some $2$-coloring of $H$.  
\end{theorem}

\begin{figure*}
\begin{center}
\fbox{
\begin{minipage}{5in}
\small
\begin{tabbing}
\textbf{LCA for \emph{Hypergraph Coloring}} \\
Preprocessing: \\
\quad 1. generate $O(\frac{\log{N}}{\log\log{N}})$ copies of $c\log{N}$-wise independent random variables in $\{0,1\}^{m}$ \\
\quad 2. generate a $\frac{1}{m^{2}}$-almost $\polylog{N}$-wise independent random ordering over $[m]$ \\
Input: a vertex $x \in V$\\
Out\=put\=: a color in \{\emph{red}, \emph{blue}\}\\
1.\> Use BFS to grow a random query tree $\mathcal{T}$ rooted at $x$ \\
2.\> Color the vertices in $\mathcal{T}$ bottom up \\
3.\> If $x$ is colored \emph{red} or \emph{blue}, return the color \\
  \>\> Else run \textbf{Phase $2$ Coloring}($x$)
\end{tabbing}
\end{minipage}
}
\end{center}
\caption{Local computation algorithm for \emph{Hypergraph Coloring}}
\label{Fig:coloring}
\end{figure*}

\begin{figure*}
\begin{center}
\fbox{
\begin{minipage}{5in}
\small
\begin{tabbing}
\textbf{Phase $2$ Coloring($x$)} \\
Input: a vertex $x \in V$ \\
Out\=put\=: a \=color in \{\emph{red}, \emph{blue}\} or \emph{FAIL}\\
1.\> Start from $\mathcal{E}(x)$ to explore $G$ in order to find the connected \\
  \>\>component $C_{1}(x)$ of \emph{survived} hyperedges around $x$\\
2.\> If the size of the component is larger than $c_{2}\log N$ \\
  \>\> Abort and return \emph{FAIL}\\
3.\> Repeat the following $O(\frac{\log{N}}{\log\log{N}})$ times and stop if
  a \emph{good} coloring is found\footnote{
  Following~\cite{RTVX11b}, let $S_{1}(x)$ be the set of surviving hyperedges in $C_{1}(x)$
after all vertices in $C_{1}(x)$ are either colored or are frozen.
Now we explore the dependency graph of $S_{1}(x)$ to find out all
the connected components. We say a Phase $2$ coloring is \emph{good} if all connected components 
  in $G|_{S_1}(x)$ have sizes at most $c_{3}\log\log{N}$, where $c_{3}$ is some absolute constant. }\\
   \>\> (a) Color all the vertices in $C_{1}(x)$ uniformly at random\\
   \>\> (b) Explore the dependency graph of $G|_{S_{1}(x)}$ \\
   \>\> (c) Check if the coloring is \emph{good} \\
4.\> If $x$ is colored in the good coloring, return that color \\
  \>\> Else run \textbf{Phase $3$ Coloring}($x$)
\end{tabbing}
\end{minipage}
}
\end{center}
\caption{Local computation algorithm for \emph{Hypergraph Coloring}:  Phase $2$}
\label{Fig:Phase2}
\end{figure*}

\begin{figure*}
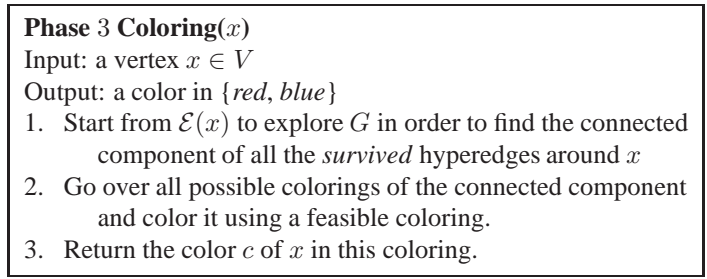

\begin{center}
\fbox{
\begin{minipage}{8in}
\small
\begin{tabbing}
\textbf{Phase $3$ Coloring($x$)} \\
Input: a vertex $x \in V$ \\
Out\=put\=: a \=color in \{\emph{red}, \emph{blue}\}\\
1.\> Start from $\mathcal{E}(x)$ to explore $G$ in order to find the connected \\
  \>\> component of all the \emph{survived} hyperedges around $x$\\
2.\> Go over all possible colorings of the connected component\\
  \>\> and color it using a feasible coloring. \\
3.\> Return the color $c$ of $x$ in this coloring.
\end{tabbing}
\end{minipage}
}
\end{center}
\caption{Local computation algorithm for \emph{Hypergraph Coloring}: Phase $3$}
\label{Fig:Phase3}
\end{figure*}

\section{LCA for \emph{Maximal Independent Set}}\label{Sec:mis}
Recall that an independent set (IS) of a graph $G$ is a subset of vertices such that no two vertices in the set
are adjacent. An independent set is called a \emph{maximal independent set}
(MIS) if it is not properly contained in any other IS. 

In~\cite{RTVX11, RTVX11b}, a two-phase LCA is presented for MIS.
For completeness, we present the pseudocode of the LCA in Appendix~\ref{Sec:code_mis}.
Let $G$ be a graph with maximum degree $d$ and suppose the queried vertex is $v$. 
In the first phase, the LCA simulates Luby's algorithm for MIS~\cite{Lub86}.
However, instead of running the parallel algorithm for $O(\log n)$ rounds 
as the original Luby's algorithm, the LCA simulates the parallel algorithm for only $O(d\log{d})$ rounds.
Following an argument of Parnas and Ron~\cite{PR07}, 
the sequential running time for simulating the parallel algorithm to determine
whether a given node is in the MIS is $d^{O(\log{d})}$.
If $v$ or any of $v$'s neighbors is put into the independent set during the first phase,
then the algorithm return ``Yes'' or ``No'', respectively.
If, on the other hand, $v$ lies in some connected component of ``surviving'' vertices
after running the first phase,
then the algorithm proceeds to the second phase algorithm, in which a simple linear-time greedy search
for an MIS of the component is performed. 
A key result proved in~\cite{RTVX11, RTVX11b} is that, after the first phase of the algorithm,
almost surely all connected components have sizes at most $O(\poly{d}\log{n})$.
Therefore the running time\footnote{Note that 
we need to run a BFS starting from $v$ to explore the connected component in which $v$ lies.
Each step of the BFS incurs a run on the explored node of the first phase LCA.} 
of the second phase is $d^{O(\log{d})}\log{n}$.

To implement such a two-phase LCA and ensure that all answers are consistent, 
we need to maintain a random tape that keeps a record of all the generated random bits
during previous runs, which implies the space complexity of the LCA is linear
in the worst case.
To see this, suppose two vertices $u$ and $v$ are connected in $G$
and $u$ is queried first.
Suppose further that the LCA runs on $u$ and finds out during the first phase that $u$ is in the IS.
If vertex $v$ is queried at some time later, we need to ensure that,
when simulating Luby's algorithm $u$ is put in the IS in some round 
(hence $v$ is deleted in the round after that).
This in turn requires that we retrieve the random bits used during the run of LCA on $u$.

A simple but crucial observation which enables us to reduce the space complexity of the LCA
for MIS is, since all the computations are ``local'',
we may replace the truly random bits used in the algorithm
with random bits of limited independence constructed in Theorem~\ref{thm:ABI86}. 

First we round the degree bound of $G$ to $\tilde{d}=2^{\lceil \log{d} \rceil}$.
Note that $d\leq \tilde{d} <2d$.
Now we can generate the probability $1/2\tilde{d}$ used
in Luby's algorithm (c.f. Figure~\ref{maxis_P1}) 
by tossing $\log{\tilde{d}}=\lceil \log{d} \rceil$ independent fair coins.

Since the second phase of the LCA is deterministic, we can therefore focus on the first phase only.
The running time of the first phase is shown to be $d^{O(\log{d})}$~\cite{RTVX11b}.
Following the notation in~\cite{RTVX11b}, for a vertex $v$ in $G$, 
let $A_{v}$ be the event that $v$ is a surviving vertex at the end of Phase 1
and let $B_{v}$ be the event that $v$ is in state ``$\perp$'' after running 
$\mathbf{MIS}_{B}$ for $O(d\log d)$ rounds, where $\mathbf{MIS}_{B}$ is a 
variant of $\mathbf{MIS}$, a subroutine of the first phase algorithm.
It was shown in~\cite{RTVX11b} that $A_{v}\subseteq B_{v}$ (Claim 4.2)
and for any subset of vertices $W$,
\begin{align*}
&\quad\Pr[\text{all vertices in $W$ are surviving vertices}]\\
&=\Pr[\cap_{v\in W} A_{v}]  \\
&\leq\Pr[\cap_{v\in W} B_{v}].
\end{align*}

Following the proof of Lemma 4.6 in~\cite{RTVX11b},
a graph $H$ on the vertices $V(G)$ 
is called a \emph{dependency graph} for $\{B_{v}\}_{v\in V(G)}$ if for all $v$ the
event $B_{v}$ is mutually independent of all $B_{u}$ such that $(u,v)\notin H$.
Let $H^3$ denote the ``distance-$3$'' graph of $H$, that is,
vertices $u$ and $v$ are connected in $H^3$ if their distance in $H$ is exactly $3$.
Let $W$ be a subset of vertices in $H^3$.
Then, since all vertices in $W$ are at least $3$-apart,
all the events $\{B_{v}\}_{v\in W}$ are mutually independent, it follows that
the probability that all vertices in $W$ are surviving vertices satisfies
\[
\Pr[\cap_{v\in W} B_{v}]=\prod_{v\in W}\Pr[B_{v}].
\]
Finally in the proof of Lemma 4.6 in~\cite{RTVX11b}, the size of $W$ is taken to be
$c_{1}\log n$ for some constant $c_{1}$ to show that, 
almost surely all connected components of surviving vertices
after Phase 1 are of sizes at most $\poly{d}\log n$.

Now we try to replace the true random bits used in the LCA in~\cite{RTVX11b}
with pseudorandom bits of limited independence.
Firstly, since the running time of the first phase is $d^{O(\log{d})}$,
hence this is also the running time of the algorithm 
if the subroutine $\mathbf{MIS}$ is replaced with $\mathbf{MIS}_{B}$.
It follows that each event $B_{v}$ depends on at most
$d^{O(\log{d})}\cdot \log{\tilde{d}}=d^{O(\log{d})}$ random bits.
Secondly, the argument we sketched in the last paragraph is still valid
as long as the events $\{B_{v}\}_{v\in H^{3}}$ are $c_{1}\log{n}$-wise independent.
Such a condition is satisfied if the random bits used in the algorithm are $k$-wise independent, 
where $k=d^{O(\log{d})}\cdot c_{1}\log{n}=d^{O(\log{d})}\log{n}$.
Note that the total number of random bits used during the first phase for all vertices is
$m=d^{O(\log{d})}\cdot n$.
Therefore all we need is a $k$-wise independent random variable in $\{0,1\}^{m}$.
By Theorem~\ref{thm:ABI86},
such random variables can be constructed with seed length 
$O(k\log{m})=d^{O(d\log d)}\log^{2}{n}$
and each random bit can be computed in time $O(k\log{m})=d^{O(d\log d)}\log^{2}{n}$.

To put everything together, we proved the following theorem regarding 
the LCA for MIS\footnote{Note that the space complexity of storing the pseudorandom bits
dominates the space complexity of local computation for each query.}:

\begin{theorem}\label{thm:MIS_main}
Let $G$ be an undirected graph with $n$ vertices and maximum degree $d$. 
Then there is a 
$d^{O(d\log d)}\log^{3}{n}, d^{O(\log{d})}\log^{2}{n}, 1/n)$-local computation algorithm 
which, on input a vertex $v$, 
decides if $v$ is in a maximal independent set.
Moreover, the algorithm will give a consistent MIS 
for every vertex in $G$.
\end{theorem}

\section*{Acknowledgments}
We would like to thank Tali Kaufman and Krzysztof Onak for enlightening discussions.

\appendix

\section{Pseudocode of the LCA for \emph{Maximal Independent Set}}\label{Sec:code_mis}
In this section we present the pseudocode of the LCA for \emph{Maximal Independent Set}.
This is taken from~\cite{RTVX11b} with slight modifications and 
we also refer interested readers to~\cite{RTVX11b} for detailed description and analysis of the algorithm.

\begin{figure*}
\begin{center}
\fbox{
\begin{minipage}{5in}
\small
\begin{tabbing}
\textsc{Maximal Independent Set: Phase $1$} \\
Input: a graph $G$ and a vertex $v \in V$\\
Out\=put\=:  \{``\=tru\=e'', \=``fa\=lse'', ``$\perp$''\}\\
   \> For $i$ from $1$ to $r=20d\log d$ \\
   \>\>(a) If $\mathbf{MIS}(v, i)=\text{``selected''}$ \\
     \>\>\> return ``true''\\
   \>\>(b) Else if $\mathbf{MIS}(v, i)=\text{``deleted''}$ \\
	  \>\>\> return ``false'' \\
   \>\>(c) Else \\
	\>\>\> return ``$\perp$''\\
\\

{$\mathbf{MIS}(v, i)$}\\
Input: a vertex $v \in V$ and a round number $i$\\
Out\=put\=:  \{``\=sel\=ect\=ed'', ``deleted'', ``$\perp$''\}\\
1. \> If $v$ is marked ``selected'' or ``deleted'' \\
   \>\> return ``selected'' or ``deleted'', respectively\\
2. \> For every $u$ in $N(v)$\\
   \>\> If $\mathbf{MIS}(u, i-1)=\text{``selected''}$ \\
   \>\>\> mark $v$ as `` deleted'' and return ``deleted''\\
3. \> $v$ chooses itself independently with probability $\frac{1}{2d}$\\
   \>\> If $v$ chooses itself\\
   \>\>\> (i) For every $u$ in $N(v)$\\
   \>\>\>\> If $u$ is marked ``$\perp$'', 
       $u$ chooses itself independently with probability $\frac{1}{2d}$ \\
   \>\>\> (ii) If $v$ has a chosen neighbor \\
   \>\>\>\> return ``$\perp$''\\
   \>\>\> (iii) Else\\
   \>\>\>\> mark $v$ as ``selected'' and return ``selected''\\
   \>\> Else \\
   \>\>\> return ``$\perp$''
  
\end{tabbing}
\end{minipage}
}
\end{center}
\caption{Local computation algorithm for MIS:  Phase $1$}
\label{maxis_P1}
\end{figure*}

\begin{figure*}
\begin{center}
\fbox{
\begin{minipage}{5in}
\small
\begin{tabbing}
{$\mathbf{MIS}_{B}(v, i)$}\\
Input: a vertex $v \in V$ and a round number $i$\\
Out\=put\=:  \{``\=pick\=ed'', ``$\perp$''\}\\
1. \> If $v$ is marked ``picked'' \\
   \>\> return ``picked''\\
2. \> $v$ chooses itself independently with probability $\frac{1}{2d}$\\
   \>\> If $v$ chooses itself\\
   \>\>\> (i) For every $u$ in $N(v)$\\
   \>\>\>\>  $u$ chooses itself independently with probability $\frac{1}{2d}$ \\
   \>\>\> (ii) If $v$ has a chosen neighbor \\
   \>\>\>\> return ``$\perp$''\\
   \>\>\> (iii) Else\\
   \>\>\>\> mark $v$ as ``picked'' and return ``picked''\\
   \>\> Else \\
   \>\>\> return ``$\perp$''
  
\end{tabbing}
\end{minipage}
}
\end{center}
\caption{Algorithm $\mathbf{MIS}_{B}$}
\label{maxis_B}
\end{figure*}

\begin{figure*}
\begin{center}
\fbox{
\begin{minipage}{5in}
\small
\begin{tabbing}
\textsc{Maximal Independent Set: Phase $2$} \\
Input: a graph $G$ and a vertex $v \in V$\\
Out\=put\=:  \{``\=tru\=e'', \=``fa\=lse''\}\\
1. \> Run BFS starting from $v$ to grow a connected component of surviving vertices \\
   \> (If a vertex $u$ is in the BFS tree and $w\in N(u)$ in $G$, then $w$ is in the BFS tree \\
     \>\> if and only if running the first phase LCA on $w$ returns ``$\perp$'')\\
2. \> (Run the greedy search algorithm on the connected component for an MIS)\\
   \> Set $S=\emptyset$ \\
   \> Scan all the vertices in the connected component in order \\
   \>\> If a vertex $u$ is not deleted \\
   \>\>\> add $u$ to $S$ \\
   \>\>\> delete all the neighbors of $u$\\
3. \> If $v\in S$ \\
   \>\> return ``true''\\
   \>\> else ``false''\\   
\end{tabbing}
\end{minipage}
}
\end{center}
\caption{Local computation algorithm for MIS:  Phase $2$}
\label{maxis_P2}
\end{figure*}

\end{document}